\documentclass[conference,letterpaper]{IEEEtran}
\usepackage[utf8]{inputenc} 
\usepackage[T1]{fontenc}
\usepackage{url}
\usepackage{ifthen}
\usepackage{cite}
\usepackage[cmex10]{amsmath} 

\usepackage{amsfonts}
\usepackage{amsmath,amsthm,amssymb}
\usepackage{mathrsfs}
\usepackage{soul}
\usepackage{bbm}
\usepackage{color}
\usepackage[dvipsnames]{xcolor}
\usepackage{graphicx}
\usepackage{tikz}
\usepackage{microtype}

\newcommand{\calD}{\mathcal{D}}

\newcommand{\calK}{\mathcal{K}}
\newcommand{\calL}{\mathcal{L}}

\newcommand{\calP}{\mathcal{P}}

\newcommand{\calU}{\mathcal{U}}

\newcommand{\calX}{\mathcal{X}}
\newcommand{\calY}{\mathcal{Y}}
\newcommand{\calZ}{\mathcal{Z}}

\newcommand{\N}{{\mathbb N}}

\newcommand{\R}{{\mathbb R}}

\newcommand{\ol}[1]{\overline{#1}}

\newcommand{\I}[1]{\mathrm{\bf 1}_{#1}}

\renewcommand{\Pr}[1]{\mathrm{P}\left(#1\right)}
\newcommand{\E}[1]{{\mathbb E}\left(#1\right)}
\newcommand{\Ex}[2]{{\mathbb E}_{#1}\left(#2\right)}

\newcommand{\bsmatrix}{\left(\begin{smallmatrix}}
\newcommand{\esmatrix}{\end{smallmatrix}\right)}

\newcommand{\eq}[1]{\begin{equation}
#1
\end{equation}}
\newcommand{\eqn}[2]{\begin{equation}
\label{#1}
#2
\end{equation}}
\newcommand{\al}[1]{\begin{align}
#1
\end{align}}
\newcommand{\aln}[1]{\begin{align}
#1
\end{align}}
\newcommand{\dsty}[1]{$\displaystyle #1$}

\newtheorem{definition}{Definition}
\newtheorem{theorem}{Theorem}

\newtheorem{proposition}{Proposition}
\newtheorem{lemma}{Lemma}

\newcounter{example}

\IEEEoverridecommandlockouts
\begin{document}

\title{On the Contractivity of Privacy Mechanisms}

\author{\IEEEauthorblockN{Mario Diaz}
\IEEEauthorblockA{Arizona State University and Harvard University\\
mdiaztor@\{asu,g.harvard\}.edu}
\and
\IEEEauthorblockN{Lalitha Sankar}
\IEEEauthorblockA{Arizona State University\\
lsankar@asu.edu}\thanks{This material is based upon work supported by the National Science Foundation under Grant No. CCF-1350914 and an ASU seed grant.}
\and
\IEEEauthorblockN{Peter Kairouz}
\IEEEauthorblockA{Stanford University\\
kairouzp@stanford.edu}}

\maketitle

\begin{abstract}
We present a novel way to compare the statistical cost of privacy mechanisms using their Dobrushin coefficient. Specifically, we provide upper and lower bounds for the Dobrushin coefficient of a privacy mechanism in terms of its maximal leakage and local differential privacy guarantees. Given the geometric nature of the Dobrushin coefficient, this approach provides some insights into the general statistical cost of these privacy guarantees. We highlight the strength of this method by applying our results to develop new bounds on the \mbox{$\ell_2$-minimax} risk in a distribution estimation setting under maximal leakage constraints. Specifically, we find its order with respect to the sample size and privacy level, allowing a quantitative comparison with the corresponding local differential privacy \mbox{$\ell_2$-minimax} risk.
\end{abstract}

\section{Introduction}

Recently, several compelling definitions for privacy have arisen, notable among them are the context-free (statistics agnostic) notion of differential privacy and the context-aware information-theoretic measures such as mutual information and maximal leakage. Context-aware approaches provide average-case privacy guarantees and allow for a range of adversarial models.

Both context-free and context-aware definitions of privacy require designing a probabilistic mapping (henceforth referred to as privacy mechanism) that satisfies the  desired privacy requirement. Despite the operational interpretations, comparing privacy leakage measures numerically does not provide much insight. One approach is to evaluate the effect of the privacy requirement on utility (i.e., the statistical cost of using the corresponding privacy mechanism). The aim of this paper is to provide a framework in which different privacy metrics can be compared in terms of their general statistical cost.

We compare different privacy mechanisms via their Dobrushin coefficient, which is equal to the contractivity coefficient of the mechanism with respect to total variation distance. Specifically, we provide upper bounds on the privacy guarantees of a local differentially private (L-DP) mechanism and a maximal leakage (MaxL) private mechanism in terms of its Dobrushin coefficient. Conversely, we provide upper bounds on the Dobrushin coefficient of any mechanism in terms of its L-DP and MaxL privacy guarantees. Since a small Dobrushin coefficient means a highly contractive mapping, the latter bounds are particularly useful to assess the cost of both aforementioned privacy guarantees for a wide range of statistical problems, including hypothesis testing and distribution estimation. More specifically, the Dobrushin coefficient can be used to provide strong data processing inequalities (SDPIs) for any $f$-divergence. Since these SDPIs are a fundamental part of many standard statistical methodologies, e.g., Le Cam's method, the Dobrushin coefficient leads to an immediate evaluation of the statistical cost of a privacy mapping.

We highlight the value of this approach by presenting new results on the \mbox{$\ell_2$-minimax} risk for a distribution estimation setting under MaxL constraints, i.e., we compute the best worst-case expected $\ell_2$-loss of a distribution estimator when using data sanitized by a privacy mechanism with specific MaxL guarantees. We show that the minimax risk order with respect to the sample size $n$ and privacy level $\alpha$ has order $(n(2^\alpha-1))^{-1}$. This is the first step to enable quantitative comparisons with the corresponding local differential private minimax risk (see, for example, \cite{pastore2016locally,duchi2017minimax,ye2017asymptotically}). The value of this approach is in exploiting the connection between the Dobrushin coefficient and Le Cam's method to obtain bounds more directly. 

This paper is organized as follows. In Section~\ref{Section:Setting} we introduce the main concepts and terminology used in this paper. In Section~\ref{Section:MainResults} we present our main results regarding the relation between L-DP or MaxL privacy guarantees and the Dobrushin coefficient of privacy mechanisms, and illustrate some of their consequences in Section~\ref{Section:Illustration}. In Section~\ref{Section:Minimax} we apply our results to study the $\ell_2$-minimax risk of a distribution estimation problem under maximal leakage constraints. The proof of our main results are provided in Appendix~\ref{Section:Proofs}.

\section{Problem Setting and Preliminaries}
\label{Section:Setting}

\subsection{Privacy Notions}

We assume that $\calX$ and $\calY$ are finite sets. A privacy mechanism is a function $W:\calX\times\calY\to\R$ such that $W(x,y)\geq0$ for all $(x,y)\in\calX\times\calY$ and, for all $x\in\calX$,
\eq{\sum_{y\in\calY} W(x,y) = 1.}
Let $\calP(\calZ)$ be the set of probability measures on a discrete set $\calZ$. Every privacy mechanism $W:\calX\times\calY\to\R$ can be identified with a mapping $W':\calP(\calX)\to\calP(\calY)$ determined by
\eq{\phantom{y\in\calY.} \quad \quad W'(P)(y) = \sum_{x\in\calX} P(x) W(x,y), \quad \quad y\in\calY.}
Note that this correspondence defines a bijection. Thus, by abuse of notation, we denote by $W$ both $W$ and $W'$.

\begin{definition}[\!\!\cite{duchi2013local}]
For $\alpha\in[0,\infty]$, a privacy mechanism $W$ is said to be $\alpha$-locally differentially private if
\eqn{eq:DefLocDiffPriv}{\max_{y\in\calY} \max_{x_1,x_2\in\calX} \frac{W(x_1,y)}{W(x_2,y)} \leq 2^\alpha.}
\end{definition}

For convenience, we define $0/0 := 1$ and $1/0:=\infty$. Following tradition, all logarithms are taken in base 2.

\begin{definition}[\!\!\cite{issa2016operational}]
\label{Def:MaxL}
For $\alpha\in[0,\infty]$, a privacy mechanism $W$ is said to be $\alpha$-MaxL private if
\eqn{eq:DefMLPriv}{\sum_{y\in\calY} \max_{x\in\calX} W(x,y) \leq 2^\alpha.}
\end{definition}

The operational significance of Def.~\ref{Def:MaxL} comes from the following result of Issa et al. \cite{issa2016operational}.

\begin{proposition}[\!\!\cite{issa2016operational}]
Let $X$ and $Y$ be random variables with support $\calX$ and $\calY$, respectively. Define
\eqn{}{\calL(X\to Y) := \sup_{U - X - Y} \log \frac{\Pr{U = \hat{U}(Y)}}{\max_{u\in\calU} P_U(u)},}
where the support of $U$ is finite but of arbitrary size, and $\hat{U}(Y)$ is the maximum a posteriori estimator. Then,
\eqn{}{\calL(X\to Y) = \log \sum_{y\in\calY} \max_{x\in\calX} P_{Y|X}(y|x).}
\end{proposition}

The Dobrushin coefficient of a probabilistic mapping is defined as follows.

\subsection{Dobrushin Coefficient}
\label{Subsection:Dobrushin}

\begin{definition}[\!\!\!\cite{dobrushin1956central}]
The Dobrushin coefficient $\ol{\alpha}(W)$ of a mapping $W:\calP(\calX)\to\calP(\calY)$ is defined by
\eqn{eq:DobrushinExpression}{\ol{\alpha}(W) = \max_{x_1,x_2\in\calX} \frac{1}{2} \sum_{y\in\calY} |W(x_1,y) - W(x_2,y)|.} 
\end{definition}

For $P,Q\in\calP(\calZ)$, their total variation distance is given by
\eq{\|P - Q\| = \frac{1}{2} \sum_{z\in\calZ} |P(z) - Q(z)|.}
In \cite{dobrushin1956central}, Dobrushin provided the next characterization for $\ol{\alpha}(W)$.

\begin{proposition}[\!\!\cite{dobrushin1956central}]
\label{Proposition:DobrushinEtaTV}
For every mapping $W:\calP(\calX)\to\calP(\calY)$,
\eqn{eq:DobrushinEtaTV}{\ol{\alpha}(W) = \sup_{\begin{smallmatrix} P_0,P_1\in\calP(\calX) \\ \|P_0 - P_1\| > 0 \end{smallmatrix}} \frac{\|W(P_0) - W(P_1)\|}{\|P_0 - P_1\|}.}
\end{proposition}

It is straightforward to verify that the right hand side of \eqref{eq:DobrushinEtaTV} is upper bounded by one for every $W$. In particular, $\ol{\alpha}(W)$ is the constant of the contractive mapping $W$.

Note that when $\ol{\alpha}(W)$ is close to zero, the output of the mapping $W$ is essentially the same for every input distribution. This cause severe degradation to the utility of any statistical methodology applied to the output of this mapping. Thus, this geometric feature of the Dobrushin coefficient offers a quantitative assessment of the statistical cost of a mapping.

\subsection{$f$-Divergences}

Another important property of the Dobrushin coefficient comes from its connection with $f$-divergences.

\begin{definition}
Let $f:\R_+\to\R$ be convex with $f(1)=0$. For $P,Q\in\calP(\calZ)$, its $f$-divergence $D_f(P \| Q)$ is determined by
\eqn{}{D_f(P \| Q) := \sum_{z\in\calZ} Q(z) f\left(\frac{P(z)}{Q(z)}\right).}
For $W:\calP(\calX)\to\calP(\calY)$, let $\eta_{f}(W)$ be defined by
\eqn{}{\eta_{f}(W) := \sup_{\begin{smallmatrix}P_0,P_1\in\calP(\calX) \\ 0<D_{f}(P_0 \| P_1)<\infty\end{smallmatrix}} \hspace{-8pt} \frac{D_{f}(W(P_0) \| W(P_1))}{D_{f}(P_0 \| P_1)}.}
\end{definition}

For example, $f(x) = x\log(x)$ leads to the KL-divergence,
\eq{D_f(P \| Q) = D_{\text{KL}}(P \| Q) := \sum_{z\in\calZ} P(z) \log\frac{P(z)}{Q(z)}.}
For this choice of $f$ it is customary to denote $\eta_{f}$ by $\eta_{\text{KL}}$. Similarly, when $f(x) = |x-1|/2$, it is easy to show that $D_f(P \| Q) = \|P - Q\|$. In this case, $\eta_{f}$ is usually denoted by $\eta_{\text{TV}}$. With this notation, Proposition~\ref{Proposition:DobrushinEtaTV} says that
\eqn{eq:DobrushinEtaTV2}{\eta_{\text{TV}}(W) = \ol{\alpha}(W).}
In the sequel, $\eta_{\text{TV}}$ and $\ol{\alpha}$ are used interchangeably. Remarkably, $\eta_{\text{TV}}$ is an upper bound for $\eta_{f}$ for every $f$, see, for example, \cite[Prop.~II.4.10]{cohen1998comparisons}.

\begin{proposition}[\!\!\!\cite{cohen1998comparisons}]
For every $f:\R_+\to\R$ convex with $f(1)=0$, it holds true that
\eqn{}{\eta_{f}(W) \leq \eta_{\text{TV}}(W).}
\end{proposition}

In particular, we have that \dsty{\eta_{\text{KL}}(W) \leq \eta_{\text{TV}}(W)}. Therefore, for every privacy mechanism $W$ and $P_0,P_1\in\calP(\calX)$,
\eqn{eq:DPIKLMaxL}{D_{\text{KL}}(W(P_0) \| W(P_1)) \leq \eta_{\text{TV}}(W) D_{\text{KL}}(P_0 \| P_1).}
This inequality will be used to derive a lower bound for the \mbox{$\ell_2$-minimax} risk in a distribution estimation setting under maximal leakage constraints. We finish this section pointing out that more refined contraction coefficients have been the subject of recent studies, see, for example, \cite{polyanskiy2016dissipation,Makur2017linear} and references therein.

\section{Main Results}
\label{Section:MainResults}

\subsection{Local Differential Privacy}
\label{MainResults:LocalDifferentialPrivacy}

The following theorem provides an upper bound for the Dobrushin coefficient of a privacy mechanism in terms of its L-DP guarantee.

\begin{theorem}
\label{Thm:LDPDobrushinCoefficient}
If a privacy mechanism $W$ is $\alpha$-locally differentially private, then
\eqn{eq:LDPDobrushin}{\eta_\text{TV}(W) \leq \frac{2^\alpha-1}{2^\alpha+1}.}
\end{theorem}

In \cite[Corollary~11]{kairouz2016extremal}, Kairouz et al. obtained Thm.~\ref{Thm:LDPDobrushinCoefficient} as a by-product of the characterization of the set of optimal mechanisms for a certain type of optimization problems under L-DP constraints. In this work, such a characterization is not required in order to establish \eqref{eq:LDPDobrushin}. This new approach provides an alternative to study privacy notions for which optimal mechanisms are not easily computable, for example, maximal leakage.

The following theorem provides a converse to Theorem~\ref{Thm:LDPDobrushinCoefficient}, i.e., provides an upper bound for the L-DP guarantee of a privacy mechanism in terms of its Dobrushin coefficient.

\begin{theorem}
\label{Thm:DobrushinCoefficientLDP}
For every privacy mechanism $W$,
\eqn{eq:DobrushinLDP}{\max_{y\in\calY} \max_{x_1,x_2\in\calX} \frac{W(x_1,y)}{W(x_2,y)} \leq 1+\frac{\eta_\text{TV}(W)}{W_*},}
where \dsty{W_* := \inf_{x\in\calX} \inf_{y\in\calY} W(x,y)}.
\end{theorem}

After some manipulations, Theorems~\ref{Thm:LDPDobrushinCoefficient} and \ref{Thm:DobrushinCoefficientLDP} imply that, for every privacy mechanism $W$,
\eq{\frac{2\eta_\text{TV}(W)}{1-\eta_\text{TV}(W)} \leq \max_{y\in\calY} \max_{x,x'\in\calX} \log\frac{W(x,y)}{W(x',y)} -1 \leq \frac{\eta_\text{TV}(W)}{W_*}.}
As shown in the sequel, both bounds can be tight. These inequalities show the connection between the local differential privacy guarantee and the contraction properties of a privacy mechanism.

\subsection{Maximal Leakage (MaxL) based Privacy}

The following theorem provides an upper bound for the Dobrushin coefficient of a privacy mechanism in terms of its MaxL privacy guarantee.

\begin{theorem}
\label{Thm:MLDobrushinCoefficient}
If a privacy mechanism $W$ is $\alpha$-MaxL private,
\eqn{eq:MaxLDobrushin}{\eta_\text{TV}(W) \leq \min\{1,2^\alpha-1\}.}
\end{theorem}

The following theorem provides a converse to Theorem~\ref{Thm:MLDobrushinCoefficient}, i.e., provides an upper bound for the MaxL guarantee of a privacy mechanism in terms of its Dobrushin coefficient.

\begin{theorem}
\label{Thm:DobrushinCoefficientMLP}
For every privacy mechanism $W$,
\eqn{eq:DobrushinMaxL}{\sum_{y\in\calY} \max_{x\in\calX} W(x,y)\leq \frac{|\calX|}{2} (1+\eta_\text{TV}(W)).}
\end{theorem}

In particular, Theorems~\ref{Thm:MLDobrushinCoefficient} and \ref{Thm:DobrushinCoefficientMLP} imply that,
\eqn{eq:ConclusionMaxL}{1+\eta_\text{TV}(W) \leq \sum_{y\in\calY} \max_{x\in\calX} W(x,y) \leq \frac{|\calX|}{2} (1+\eta_\text{TV}(W)),}
for every privacy mechanism $W$. Note that when $|\calX| = 2$, \eqref{eq:ConclusionMaxL} reduces to an equality. The inequalities in \eqref{eq:ConclusionMaxL} show the intrinsic relation between the MaxL privacy guarantee of a privacy mechanism and its Dobrushin coefficient.

\section{Illustration}
\label{Section:Illustration}

\subsection{Local Differential Privacy}
\label{Section:IllustrationLDP}

Assume that $\calX=\calY=\{1,\ldots,k\}$. Consider the randomized response mechanism $W_{k,\alpha}:\calP(\calX)\to\calP(\calY)$ defined by
\eq{W_{k,\alpha}(x,y) = \frac{2^\alpha-1}{2^\alpha+k-1} \mathbbm{1}_{\{x=y\}} + \frac{1}{2^\alpha+k-1}.}
Note that $W_{k,\alpha}$ is $\alpha$-locally differentially private. A straightforward computation shows that, for all $P_0,P_1\in\calP(\calX)$,
\eqn{eq:etaWkalpha}{\|W_{k,\alpha}(P_0) - W_{k,\alpha}(P_1)\| = \frac{2^\alpha-1}{2^\alpha+k-1} \|P_0 - P_1\|,}
and thus \dsty{\eta_{\text{TV}}(W_{k,\alpha}) = \frac{2^\alpha-1}{2^\alpha+k-1}}. The fact that \eqref{eq:etaWkalpha} holds for all $P_0,P_1\in\calP(\calX)$ shows that the randomized response mechanism contracts the space $\calP(\calX)$ in a uniform way. This is a desirable property in terms of statistical utility. For the binary setting,
\eq{\eta_{\text{TV}}(W_{2,\alpha}) = \frac{2^\alpha-1}{2^\alpha+1},}
thus illustrating a case where the bound in Theorem~\ref{Thm:LDPDobrushinCoefficient} is tight. Also, note that $(W_{2,\alpha})_* = (2^\alpha+1)^{-1}$. Thus implying that the bound in Theorem~\ref{Thm:DobrushinCoefficientLDP} is also tight.

\subsection{Maximal Leakage Privacy}
\label{Section:IllustrationMaxL}

Let $\calX=\calY=\{0,1\}$. For $\alpha\leq1$, put $\zeta_\alpha=2-2^\alpha$. Consider the privacy mechanism
\eq{Z_\alpha := \left(\begin{matrix} 1 - \zeta_\alpha & \zeta_\alpha \\ 0 & 1 \end{matrix}\right).}
Clearly, $Z_\alpha$ is $\alpha$-MaxL private. Furthermore, a straightforward computation shows that, for all $P_0,P_1\in\calP(\{0,1\})$,
\eq{\|Z_\alpha(P_0) - Z_\alpha(P_1)\| = (1-\zeta_\alpha) \|P_0 - P_1\|.}
Hence, $\eta_{\text{TV}}(Z_\alpha) = 1- \zeta_\alpha = 2^\alpha -1$. This type of {\it $Z$-channel mechanism} has proved to be optimal for some problems under maximal leakage constraints; see, for example, Liao et al. \cite{liao2017hypothesis} in the context of hypothesis testing.

\section{Maximal Leakage $\ell_2$-Minimax Risk}
\label{Section:Minimax}

In this section we analyze the maximal leakage $\ell_2$-minimax risk, adapting the notion of local differential private minimax risk in \cite{duchi2017minimax}. In particular, we provide lower and upper bounds for the $\alpha$-MaxL $\ell_2$-minimax risk with matching orders with respect to the sample size and privacy level. A detailed treatment of the non-private version of this problem can be found in \cite{kamath2015learning}.

For probability distributions $P,Q\in\calP(\calZ)$, we let $\|P - Q\|_2$ be their $\ell_2$-distance, i.e.,
\eq{\|P - Q\|_2^2 = \sum_{z\in\calZ} (P(z)-Q(z))^2.}
Let $\Delta_k$ be the set of distributions over $[k]:=\{1,\ldots,k\}$. By definition, the $\ell_2$-minimax risk of a privacy mechanism $W:\Delta_k \to \Delta_m$ is given by
\eqn{eq:DefMinimaxW}{r_{k,n}^{\|\cdot\|_2^2}(W) := \inf_{\hat{P}} \sup_{P\in\Delta_k} \Ex{W,P}{\|\hat{P}(Y^n) - P\|_2^2},}
where the infimum is over all the estimators $\hat{P}:[m]^n\to\Delta_k$ and $\Ex{W,P}{Z}$ denotes expectation of $Z$ when $Y_1,\ldots,Y_n$ are i.i.d. with distribution $W(P)$. For $\alpha\in[0,\infty]$, we define the $\alpha$-MaxL $\ell_2$-minimax risk as
\eqn{eq:DefMinimax}{r_{\alpha,k,n}^{\|\cdot\|_2^2} = \inf_{W\in\calD_\alpha} r_{k,n}^{\|\cdot\|_2^2}(W) ,}
where $\calD_\alpha$ denotes the set of all privacy mechanisms which are $\alpha$-MaxL private. The next proposition provides a lower bound for $r_{\alpha,k,n}^{\|\cdot\|_2^2}$. Our proof, which is provided at the end of this section, relies on Le Cam's method, as in Ye and Barg \cite[pp.~26--27]{ye2017asymptotically}, and \eqref{eq:DPIKLMaxL}.

\begin{proposition}
\label{Prop:DistributionEstimationMinimaxMaxL}
Let $k\in\N$ and $\alpha>0$ be given. There exists $N=N(k,\alpha)$ such that, for all $n>N$,
\eqn{eq:LBl2}{r_{\alpha,k,n}^{\|\cdot\|_2^2} \geq \frac{1}{16n(2^\alpha-1)}.}
\end{proposition}

By providing an achievability scheme, we provide an upper bound for $r_{\alpha,k,n}^{\|\cdot\|_2^2}$.

\begin{proposition}
\label{Prop:DistributionEstimationMinimaxMaxL2}
Let $k\in\N$ and $\alpha>0$ be given. If $2^\alpha\leq k$, then, for all $n\in\N$,
\eqn{eq:UBl2}{r_{\alpha,k,n}^{\|\cdot\|_2^2} \leq \frac{k-1}{n(2^\alpha-1)}.}
\end{proposition}

Note that when $2^\alpha > k$, the constraint in Def.~\ref{Def:MaxL} becomes vacuous. By Propositions~\ref{Prop:DistributionEstimationMinimaxMaxL} and \ref{Prop:DistributionEstimationMinimaxMaxL2},
\eq{r_{\alpha,k,n}^{\|\cdot\|_2^2} = \Theta_k\left(\frac{1}{n(2^\alpha-1)}\right).}
The notation $f(n,\alpha) = \Theta_k(g(n,\alpha))$ denotes that there exists $N=N(k,\alpha)$ such that, for all $n>N$,
\eq{C_1 g(n,\alpha) \leq f(n,\alpha) \leq C_2 g(n,\alpha),}
where $C_1,C_2>0$ are constants depending only on $k$. In particular, the MaxL $\ell_2$-minimax risk is smaller than its L-DP counterpart, see \cite[Thm.~II.5]{ye2017asymptotically},
\eq{r_{\alpha,k,n}^{\|\cdot\|_2^2} = \Theta_k\left(\frac{2^\alpha}{n(2^\alpha-1)^2}\right).}

\begin{proof}[\bf Proof of Proposition~\ref{Prop:DistributionEstimationMinimaxMaxL}]
Let $P_0\in\Delta_k$ be the uniform distribution over $[k$], i.e., $P_0(x) = k^{-1}$ for all $x\in[k]$. For a given vector $u\in\R^k$ such that $\sum_x u_x = 0$ and $\sum_x u_x^2 = 1$, let $P_1\in\Delta_k$ be the distribution determined by
\eq{\phantom{x\in\calX.} \quad P_1(x) = \frac{1}{k} + \frac{u_x}{\sqrt{n(2^\alpha-1)}}, \quad x\in[k].}
Note that if $n\geq k^2/(2^\alpha-1)$, then $P_1$ indeed defines a probability distribution. A direct computation shows that
\eqn{eq:ProofMinimaxDistanceP01}{\| P_0 - P_1 \|_2 = \frac{1}{\sqrt{n(2^\alpha-1)}}.}

Fix an $\alpha$-MaxL private mechanism $W:\Delta_k\to\Delta_m$. For a given estimator $\hat{P}$, let 
\eq{S(\hat{P}) := \frac{1}{2} \sum_{i=0,1} \Ex{W,P_i}{\|\hat{P}(Y^n) - P_i\|_2^2}.}
Let \dsty{\calK_{\hat{P}} := \{y^n\in[m]^n : \|\hat{P}(y^n) - P_0\|_2 \geq \|\hat{P}(y^n) - P_1\|_2\}}. By \eqref{eq:ProofMinimaxDistanceP01} and the triangle inequality, for $y^n\in\calK_{\hat{P}}$,
\eq{\|\hat{P}(y^n) - P_0\|_2 \geq \frac{1}{2} \| P_0 - P_1 \|_2 = \frac{1}{2\sqrt{n(2^\alpha-1)}},}
and, for $y^n\in\calK_{\hat{P}}^c$,
\eq{\|\hat{P}(y^n) - P_1\|_2 \geq \frac{1}{2} \| P_0 - P_1 \|_2 = \frac{1}{2\sqrt{n(2^\alpha-1)}}.}
Also, observe that
\al{
\Ex{W,P_0}{\|\hat{P}(Y^n) - P_0\|_2^2} &\geq \Ex{W,P_0}{\|\hat{P}(Y^n) - P_0\|_2^2 \I{\calK_{\hat{P}}}}\\
&\geq \frac{1}{4n(2^\alpha-1)} \Ex{W,P_0}{\I{\calK_{\hat{P}}}}\\
&= \frac{1}{4n(2^\alpha-1)} W(P_0)^n(\calK_{\hat{P}}),
}
where $W(P_0)^n(\calK_{\hat{P}})$ denotes the measure of $\calK_{\hat{P}}$ with respect to the $n$-fold tensor product measure $W(P_0) \otimes \cdots \otimes W(P_0)$. Using a similar argument, we obtain that
\eq{S(\hat{P}) \geq \frac{1}{8n(2^\alpha-1)} \big[W(P_0)^n(\calK_{\hat{P}}) + W(P_1)^n(\calK_{\hat{P}}^c)\big].}
Recall that $P(E)+Q(E^c) \geq 1 - \|P - Q\|$ for every event $E$. In particular,
\al{
S(\hat{P}) &\geq \frac{1}{8n(2^\alpha-1)}(1 - \|W(P_0)^n - W(P_1)^n\|)\\
&\geq \frac{1}{8n(2^\alpha-1)} \left(1 - \sqrt{\frac{n}{2} D_{\text{KL}}(W(P_1) \| W(P_0))}\right),
}
where the last inequality follows from Pinsker's inequality and the tensorization property of the KL-divergence. By \eqref{eq:DPIKLMaxL} and Theorem~\ref{Thm:MLDobrushinCoefficient}, we obtain that
\eq{S(\hat{P}) \geq \frac{1}{8n(2^\alpha-1)} \left(1 - \sqrt{\frac{n(2^\alpha-1)}{2} D_{\text{KL}}(P_1 \| P_0)}\right).}
A Taylor series expansion argument shows that for $n$ large enough $n(2^\alpha-1)D_{\text{KL}}(P_1 \| P_0) \leq 1$, and hence
\eqn{eq:ProofMinimaxFinalS}{S(\hat{P}) \geq \frac{1}{16n(2^\alpha-1)}.}
Note that \dsty{\sup_{P\in\Delta_k} \Ex{W,P}{\|\hat{P}(Y^n) - P\|_2^2} \geq S(\hat{P})}, thus
\eq{r_{k,n}^{\|\cdot\|_2^2}(W) = \inf_{\hat{P}} \sup_{P\in\Delta_k} \Ex{W,P}{\|\hat{P}(Y^n) - P\|_2^2} \geq \inf_{\hat{P}} S(\hat{P}).}
By \eqref{eq:ProofMinimaxFinalS}, we obtain that
\eq{r_{k,n}^{\|\cdot\|_2^2}(W) \geq \frac{1}{16n(2^\alpha-1)}.}
Since this inequality holds for any given $\alpha$-MaxL private mechanism $W$, the result follows.
\end{proof}

\begin{proof}[\bf Proof of Proposition~\ref{Prop:DistributionEstimationMinimaxMaxL2}]
Let \dsty{\lambda := \frac{2^\alpha-1}{k-1}}. Consider the mapping $W:\Delta_k \to \Delta_{k+1}$ given by
\eq{W = \left(\begin{matrix} \lambda & & & 1- \lambda\\ & \ddots & & \vdots \\& &  \lambda & 1-\lambda \end{matrix}\right).}
It is immediate to verify that $W$ is $\alpha$-MaxL private. Hence, \eqref{eq:DefMinimax} readily implies that
\eqn{eq:ProofAchievabilityInq1}{r_{\alpha,k,n}^{\|\cdot\|_2^2} \leq r_{k,n}^{\|\cdot\|_2^2}(W).}
Also, let $\hat{P}$ be the estimator determined by
\eq{\phantom{x\in[k].} \quad \quad \hat{P}(x) = \frac{k-1}{2^\alpha-1} \frac{1}{n} \sum_{i=1}^n \I{Y_i=x}, \quad \quad x\in[k].}
Hence, by \eqref{eq:DefMinimaxW} and \eqref{eq:ProofAchievabilityInq1},
\eqn{eq:ProofAchievabilityInq2}{r_{\alpha,k,n}^{\|\cdot\|_2^2} \leq \sup_{P\in\Delta_k} \Ex{W,P}{\|\hat{P}(Y^n) - P\|_2^2}.}
We now estimate the right hand side term of \eqref{eq:ProofAchievabilityInq2}.

For a given $P\in\Delta_k$, we let $Q = W(P)$ be the common distribution of $Y_1,\ldots,Y_n$. Note that $Q(x) = \lambda P(x)$ for all $x\in[k]$. In particular,
\eq{\|\hat{P}(Y^n) - P\|_2^2 = \frac{1}{(n \lambda)^2} \sum_{x\in[k]} \left(\sum_{i=1}^n \big[\I{Y_i=x} - Q(x)\big]\right)^2.}
Using the fact that $Q(x) = \lambda P(x)$ for all $x\in[k]$, we obtain
\al{
\E{\|\hat{P}(Y^n) - P\|_2^2} &= \frac{1}{n\lambda^2} \sum_{x\in[k]} Q(x)(1-Q(x))\\
&= \frac{1}{n\lambda} \sum_{x\in[k]} P(x)(1-\lambda P(x)).
}
Since $\sum_{x} P(x)(1-\lambda P(x)) \leq 1$, we obtain that
\eqn{eq:ProofAchievabilityInq3}{\E{\|\hat{P}(Y^n) - P\|_2^2}  \leq \frac{1}{n\lambda} = \frac{k-1}{n(2^\alpha-1)}}
Since \eqref{eq:ProofAchievabilityInq3} holds for every $P\in\Delta_k$, by \eqref{eq:ProofAchievabilityInq2} the result follows.
\end{proof}

\section{Concluding Remarks}

We have introduced a novel way to compare the statistical costs of any privacy mechanism via the Dobrushin coefficient. A significant advantage of this approach is that it eliminates the need to compute the precise mechanism for any privacy definition, which is often times difficult to obtain in closed form. Many questions remain to be addressed including tighter bounds for distribution estimation under MaxL constraints as well as application to other statistical problems with different privacy requirements.

\section*{Acknowledgement}

The authors would like to thank Dr. Ibrahim Issa for many useful discussions at the earlier stages of this work.

\appendices
\section{Proofs of the Main Results}
\label{Section:Proofs}

\subsection{Local Differential Privacy Results}

The proof of Theorem~\ref{Thm:LDPDobrushinCoefficient} is based on the following elementary observation.

\begin{lemma}
\label{Lemma:LDP}
If a privacy mechanism $W$ is $\alpha$-locally differentially private, then, for all $x_1,x_2\in\calX$ and $y\in\calY$,
\eqn{eq:InequalityDifferenceMaxMin}{\frac{|W(x_1,y)-W(x_2,y)|}{W(x_1,y) + W(x_2,y)} \leq \frac{2^\alpha - 1}{2^\alpha + 1}.}
\end{lemma}

\begin{proof}
Without loss of generality, assume that $W(x_1,y) \geq W(x_2,y)$. Note that \eqref{eq:InequalityDifferenceMaxMin} holds true if and only if
\eq{\text{\footnotesize $(2^\alpha + 1) (W(x_1,y) - W(x_2,y)) \leq (2^\alpha-1) (W(x_1,y) + W(x_2,y))$}.}
As the later holds if and only if \dsty{W(x_1,y) \leq 2^\alpha W(x_2,y)}, and $W$ is $\alpha$-locally differentially private, the result follows.
\end{proof}

\begin{proof}[\bf Proof of Theorem~\ref{Thm:LDPDobrushinCoefficient}]
For $x_1,x_2\in\calX$, we define
\eq{S(x_1,x_2) := \frac{1}{2} \sum_{y\in\calY} |W(x_1,y) - W(x_2,y)|.}
Note that for every $x_1,x_2\in\calX$, $S(x_1,x_2)$ equals
\eq{\frac{1}{2} \sum_{y\in\calY} \frac{|W(x_1,y)-W(x_2,y)|}{W(x_1,y) + W(x_2,y)} \big[W(x_1,y) + W(x_2,y)\big].}
By the Lemma~\ref{Lemma:LDP}, we have that
\al{
S(x_1,x_2) &\leq \frac{2^\alpha-1}{2(2^\alpha+1)} \sum_{y\in\calY} \big[W(x_1,y) + W(x_2,y)\big]\\
&= \frac{2^\alpha - 1}{2^\alpha + 1}.
}
By \eqref{eq:DobrushinExpression}, we obtain that \dsty{\eta_\text{TV}(W) \leq \frac{2^\alpha - 1}{2^\alpha + 1}}, as required.
\end{proof}

\begin{proof}[\bf Proof of Theorem~\ref{Thm:DobrushinCoefficientLDP}]
Without loss of generality, assume that
\eq{\max_{y\in\calY} \max_{x,x'\in\calX} \log\frac{W(x,y)}{W(x',y)} = \frac{W(x_2,y_0)}{W(x_1,y_0)}}
for some $x_1,x_2\in\calX$ and $y_0\in\calY$. For ease of notation, let $\eta = \eta_{\text{TV}}(W)$. By \eqref{eq:DobrushinExpression},
\al{
\eta &= \max_{x,x'} \frac{1}{2} \sum_{y\in\calY} |W(x,y) - W(x',y)|\\
&\geq \frac{1}{2} \sum_{y\in\calY} |W(x_2,y) - W(x_1,y)|.
}
In particular, we have that
\eq{W(x_2,y_0) - W(x_1,y_0) \leq 2\eta - \sum_{y\neq y_0} |W(x_2,y) - W(x_1,y)|.}
An elementary computation shows that
\al{
\sum_{y\neq y_0} |W(x_2,y) - W(x_1,y)| &\geq \sum_{y\neq y_0} W(x_1,y) - W(x_2,y)\\
&= W(x_2,y_0) - W(x_1,y_0),
}
and hence
\eq{W(x_2,y_0) - W(x_1,y_0) \leq 2\eta - (W(x_2,y) - W(x_1,y)),}
i.e., $W(x_2,y_0) - W(x_1,y_0) \leq \eta$. Therefore,
\eq{\frac{W(x_2,y_0)}{W(x_1,y_0)} \leq 1 + \frac{\eta_{\text{TV}}(W)}{W(x_1,y_0)} \leq 1 + \frac{\eta_{\text{TV}}(W)}{W_*}.}
The result follows.
\end{proof}

\subsection{Maximal Leakage Privacy Results}

\begin{proof}[\bf Proof of Thm.~\ref{Thm:MLDobrushinCoefficient}]
For $x_1,x_2\in\calX$, we define
\eq{S(x_1,x_2) := \frac{1}{2} \sum_{y\in\calY} |W(x_1,y) - W(x_2,y)|.}
Fix $x_1,x_2\in\calX$. Let
\eq{\calY_+ = \{y\in\calY : W(x_1,y) \geq W(x_2,y)\}}
and $\calY_- = \calY\setminus\calY_+$.
In particular, we have that
\al{
\nonumber S(x_1,x_2) =& \frac{1}{2} \sum_{y\in\calY_+} \big[W(x_1,y)-W(x_2,y)\big]\\
& \quad \quad + \frac{1}{2} \sum_{y\in\calY_-} \big[W(x_2,y) - W(x_1,y)\big].
}
Since, for every $x\in\calX$,
\eqn{eq:SumW}{\sum_{y\in\calY_+} W(x,y) + \sum_{y\in\calY_-} W(x,y) = 1,}
we have that
\al{
S(x_1,x_2) &= \sum_{y\in\calY_+} W(x_1,y) + \sum_{y\in\calY_-} W(x_2,y) - 1\\
&= \sum_{y\in\calY} \max\{W(x_1,y),W(x_2,y)\} - 1,}
where the last equality follows from the definition of $\calY_{\pm}$. Note that for all $y\in\calY$,
\eq{\max\{W(x_1,y),W(x_2,y)\} \leq \max_x W(x,y).}
Hence, \dsty{S(x_1,x_2) \leq \sum_{y\in\calY} \max_{x\in\calX} W(x,y) -1 \leq 2^\alpha -1},
where the last inequality follows as $W$ is \mbox{$\alpha$-MaxL} private. By \eqref{eq:DobrushinExpression}, we conclude that
\eq{\eta_\text{TV}(W) = \max_{x_1,x_2\in\calX} S(x_1,x_2) \leq 2^\alpha -1.}
Since $\eta_\text{TV}(W) \leq 1$ for every mapping $W:\calP(\calX)\to\calP(\calY)$, the result follows.
\end{proof}

Before proceeding with the proof of Theorem~\ref{Thm:DobrushinCoefficientMLP}, we prove the following elementary lemma.

\begin{lemma}
\label{Lemma:Means}
If $k>1$ and $a_1,\ldots,a_k$ are non-negative real numbers, then there exist $i_1 \neq i_2$ such that
\eqn{}{\frac{a_1+\cdots+a_k}{k} \leq \frac{a_{i_1}+a_{i_2}}{2}.}
\end{lemma}

\noindent{\it Proof.} Let $s = a_1+\cdots+a_k$. In order to reach contradiction, assume that \dsty{\frac{a_{i_1}+a_{i_2}}{2} < \frac{s}{k}} for all $i_1 \neq i_2$. In this case,
\eqn{eq:LemmaInq}{\sum_{i_1 \neq i_2} \frac{a_{i_1}+a_{i_2}}{2} < \sum_{i_1 \neq i_2} \frac{s}{k} = (k-1)s,}
where the last equality follows from the fact that
\eq{|\{(i_1,i_2) : i_1 \neq i_2\}| = k(k-1).}
In a similar way, for $i\in\{1,\ldots,k\}$,
\eq{|\{(i_1,i_2) : i_1 \neq i_2,i_1=i \text{ or } i_2=i\}| = 2(k-1).}
In particular, we have that
\eq{\sum_{i_1 \neq i_2} \frac{a_{i_1}+a_{i_2}}{2} = \sum_{i} \sum_{\begin{smallmatrix}i_1 \neq i_2\\i_1=i \text{ or } i_2=i\end{smallmatrix}} \frac{a_i}{2} = (k-1)s.}
By \eqref{eq:LemmaInq}, we conclude that $(k-1)s < (k-1)s$. Contradiction. \qed

\begin{proof}[\bf Proof of Theorem~\ref{Thm:DobrushinCoefficientMLP}]
For each $y\in\calY$, choose $x^{(y)}\in\calX$ such that \dsty{W(x^{(y)},y) = \max_{x\in\calX} W(x,y)}. In particular,
\eq{\sum_{y\in\calY} \max_{x\in\calX} W(x,y) = \sum_{y\in\calY} W(x^{(y)},y).}
For each $x\in\calX$, let $\calY_x := \{ y\in\calY : x^{(y)} = x\}$. Note that,
\al{
\sum_{y\in\calY} \max_{x\in\calX} W(x,y) &= \sum_{y\in\calY} W(x^{(y)},y)\\
&= \sum_{x\in\calX} \sum_{y\in\calY_x} W(x^{(y)},y),
}
where the last equality uses the fact that $\{\calY_x : x\in\calX\}$ is a partition of $\calY$. Note that $W(x^{(y)},y) = W(x,y)$ for every $y\in\calY_x$, thus
\eqn{eq:ProofThm4ChangeVariable}{\sum_{y\in\calY} \max_{x\in\calX} W(x,y) = \sum_{x\in\calX} \sum_{y\in\calY_x} W(x,y).}
By Lemma~\ref{Lemma:Means}, \eqref{eq:ProofThm4ChangeVariable} implies that there exist $x_1 \neq x_2$ such that
\eq{\frac{2}{|\calX|} \sum_{y\in\calY} \max_{x\in\calX} W(x,y) \leq \sum_{y\in\calY_{x_1}} W(x_1,y) + \sum_{y\in\calY_{x_2}} W(x_2,y).}
Note that for all $y\in\calY_{x_1}$, $W(x_1,y) = W(x^{(y)},y) \geq W(x_2,y)$ and, in particular,
\eq{W(x_1,y) = \max\{W(x_1,y),W(x_2,y)\}.}
Also, $W(x_2,y) = \max\{W(x_1,y),W(x_2,y)\}$ for all $y\in\calY_{x_2}$. Altogether, we have that
\aln{
\nonumber \frac{2}{|\calX|} \sum_{y\in\calY} \max_{x\in\calX} W(x,y) &\leq \hspace{-3pt} \sum_{y\in\calY_{x_1}\cup\calY_{x_2}} \hspace{-3pt} \max\{W(x_1,y),W(x_2,y)\}\\
\label{eq:ProofThm4InqMaxLMax}&\leq \sum_{y\in\calY} \max\{W(x_1,y),W(x_2,y)\},
}
where we used the fact that $\calY_{x_1} \cap \calY_{x_2} = \emptyset$ and $\calY_{x_1} \cup \calY_{x_2} \subset \calY$. Recall that, by \eqref{eq:DobrushinExpression},
\eq{\eta_{\text{TV}}(W) = \max_{x,x'\in\calX} \frac{1}{2} \sum_{y\in\calY} |W(x,y) - W(x',y)|.}
In particular, we have that
\al{
\eta_{\text{TV}}(W) \geq& \frac{1}{2} \sum_{y\in\calY} |W(x_1,y) - W(x_2,y)|\\
=& \frac{1}{2}\sum_{y\in\calY_+} \big[W(x_1,y)-W(x_2,y)\big]\\
& \quad \quad + \frac{1}{2}\sum_{y\in\calY_-} \big[W(x_2,y)-W(x_1,y)\big],
}
where, as before, $\calY_+ = \{y\in\calY : W(x_1,y) \geq W(x_2,y)\}$ and $\calY_-  = \calY\setminus\calY_+$. By \eqref{eq:SumW} we obtain that
\eq{\eta_{\text{TV}}(W) \geq \sum_{y\in\calY_+} W(x_1,y) + \sum_{y\in\calY_-} W(x_2,y) - 1.}
By definition of $\calY_{\pm}$, we conclude that
\eqn{eq:ProofThm4LastInq}{1+\eta_{\text{TV}}(W) \geq \sum_{y\in\calY} \max\{W(x_1,y),W(x_2,y)\}.}
By \eqref{eq:ProofThm4LastInq} and \eqref{eq:ProofThm4InqMaxLMax}, we conclude that
\eq{1+\eta_{\text{TV}}(W)\geq \frac{2}{|\calX|} \sum_{y\in\calY} \max_{x\in\calX} W(x,y),}
as we wanted to show.
\end{proof}

\bibliographystyle{IEEEtran}
\bibliography{ContractivityPrivacyMechanismsFullVersion}
\end{document}